\documentclass[a4paper,12pt]{article}
\usepackage{amsfonts}
\usepackage{latexsym}
\usepackage{amsmath}
\usepackage{amssymb}
\usepackage{amssymb}
\usepackage{slashed}
\usepackage{upgreek}
\usepackage{mathrsfs}
\usepackage{bbm}

\usepackage{amsthm}

\theoremstyle{remark}

\usepackage[utf8]{inputenc}
\usepackage[english]{babel}

\theoremstyle{definition}
\newtheorem{definition}{Definition}[section]

\newtheorem{theorem}{Theorem}[section]
\newtheorem{corollary}{Corollary}[theorem]

\newtheorem{prop}{{\bf Proposition}}[section]

\usepackage{scalerel,stackengine}
\stackMath
\newcommand\reallywidehat[1]{%
\savestack{\tmpbox}{\stretchto{%
  \scaleto{%
    \scalerel*[\widthof{\ensuremath{#1}}]{\kern-.6pt\bigwedge\kern-.6pt}%
    {\rule[-\textheight/2]{1ex}{\textheight}}%WIDTH-LIMITED BIG WEDGE
  }{\textheight}%
}{0.5ex}}%
\stackon[1pt]{#1}{\tmpbox}%
}

\newcommand*{\defeq}{\mathrel{\vcenter{\baselineskip0.5ex \lineskiplimit0pt
                     \hbox{\scriptsize.}\hbox{\scriptsize.}}}%
                     =}

\allowdisplaybreaks

\usepackage{relsize}
\usepackage{graphicx}
\usepackage{comment}

\renewcommand{\thefootnote}{\fnsymbol{footnote}}

\def\appendix#1{\addtocounter{section}{1}\setcounter{equation}{0}
\renewcommand{\thesection}{\Alph{section}}
\section*{Appendix \thesection\protect\indent \parbox[t]{11.15cm}{#1}}
\addcontentsline{toc}{section}{Appendix \thesection\ \ \ #1}}

\font\mybb=msbm10 at 11pt

\def\bb#1{\hbox{\mybb#1}}

\def\bR {\bb{R}}

\def\bC {\bb{C}}

\newcommand\fe{{\mathbf e}}

\newcommand{\bea}{\begin{eqnarray}}
\newcommand{\eea}{\end{eqnarray}}

\usepackage[usenames,dvipsnames,svgnames,table]{xcolor}	% AH -- REMOVE BEFORE UPLOAD
\usepackage[normalem]{ulem}				% AH -- REMOVE BEFORE UPLOAD
\usepackage{microtype}
\usepackage[colorlinks, citecolor=blue, linkcolor=blue, urlcolor=Maroon, filecolor=Maroon, linktocpage=true]{hyperref} % AH

\usepackage{chngcntr}
\counterwithout{equation}{section}

\begin{document}

\begin{center}
%\today
\vspace*{-1.0cm}
\begin{flushright}
%\normalsize{\texttt{ZMP-HH/17-24}}\\
\end{flushright}
%\hfill hep-th/yymmnnn \\
%\hfill UB-ECM-PF-06-43 \\
%\hfill DMUS--MP--13/06 \\

%\vspace{2.0cm} {\Large \bf Special geometric structures on Lorentzian manifolds  and their boundaries} \\[.2cm]

\vspace{2.0cm} {\Large \bf Geometry and symmetries of null G-structures} \\[.2cm]

\vskip 2cm
  G.  Papadopoulos$^*$\footnotetext{*On study leave from the Department of Mathematics, King's College London, Strand, London WC2R 2LS, UK.}
\\
\vskip .6cm

%\begin{small}
%${}^1$ \textit{Department of Mathematics, King's College London
%\\
%Strand, London WC2R 2LS, UK}\\
%\texttt{sebastian.lautz@kcl.ac.uk}\\
%\end{small}
%\vskip0.5cm

\begin{small}
\textit{Laboratoire de Physique Theorique de l'Ecole Normale Superi\'eure
\\
24 rue Lhomond
\\
75231 Paris Cedex, France}\\
\texttt{george.papadopoulos@kcl.ac.uk}
\end{small}
\\*[.6cm]

\end{center}

\vskip 2.5 cm

\begin{abstract}
\noindent
We present a definition of  null G-structures on Lorentzian  manifolds  and investigate  their geometric properties.  This definition includes   the Robinson structure on 4-dimensional black holes as well as the null structures that appear
in all supersymmetric solutions of supergravity theories. We also identify the  induced geometry
on some null hypersurfaces   and  that on the  orbit spaces of  null geodesic congruences in such Lorentzian manifolds.
We give the algebra of diffeomorphisms that preserves a null G-structure and demonstrate that in some cases it interpolates between the  BMS  algebra of an asymptotically flat spacetime and the Lorentz symmetry algebra of  a Killing horizon.

%\vskip 1cm

%{\small Keywords: spacetime geometry; black holes; special Lorentzian structures; G-structures}
\end{abstract}

\newpage

\renewcommand{\thefootnote}{\arabic{footnote}}
%\tableofcontents

%%%%%%%%%%%%%%%%%%%%%%%%%%%%%%%%%%%%%%%%%%%%%%%%%%%%%%%%%%%%%%%%%%%%%%%%%%
%\setcounter{section}{0}\setcounter{equation}{0}
%%%%%%%%%%%%%%%%%%%%%%%%%%%%%%%%%%%%%%%%%%%%%%%%%%%%%%%%%%%%%%%%%%%%%%%%%%
\section{Introduction}

 It has been known for sometime that many    solutions of 4-dimensional gravity theories admit a Lorentzian holomorphic structure.
Such solutions include black hole solutions, like Schwarzchild, Reissner–Nordstr\"om and Kerr, and other solutions like G\"odel.
Such a structure has originally been  introduced by Robinson and has been used to find new solutions like the Robinson-Trautman  solution.
Apart from the holomorphic structure another characteristic of Robinson manifolds is the existence of a nowhere vanishing
null vector field $X$, which  may not be Killing, whose integral curves are null geodesics.
It can also be shown   that the orbit space of the null geodesic congruence generated by $X$
admits a Cauchy-Riemann (CR) structure, for a review and a  proof of some of the above statements  see  \cite{nt}.  CR is the structure inherited on  hypersurfaces in complex manifolds from the ambient complex structure of the whole space and has extensively been   investigated in the literature, see e.g. \cite{ntx}-\cite{gt}.

Many of the features of Robinson manifolds like the holomorphic structure and the existence of a null vector field are reminiscent of the properties   of some supersymmetric backgrounds, see \cite{ugjggprev} for a review and references within. Moreover
the systematic exploration of supersymmetric backgrounds has revealed many new  Lorentzian geometries.  As we shall demonstrate  some of them generalize that of Robinson manifolds.  Nevertheless both the Robinson structure   and those that appear on supersymmetric solutions,  particularly those for which one of bilinears of a Killing spinor is a null vector field, belong to the same family.  This is because of  the existence of a null nowhere vanishing vector field $X$ on the spacetime. We refer to all of these structures as null G-structures\footnote{Typically null structures have been investigated on hypersurfaces of a spacetime with a null normal vector field. The null structures investigated here  will be on the whole spacetime.}.

The investigation of geometry  of  supersymmetric backgrounds has given a new perspective  in the description of geometry of null structures. This is because it is based on the properties of {\it globally defined} fundamental forms on the spacetime.  As an approach it is similar to that taken for the description of geometry  of n-dimensional Riemannian manifolds with a G-structure for $G\subseteq O(n)$. However there are some  key differences. One is  that for null G-structures  all the fundamental forms   are null.  Another is that the null structure is a conformally invariant concept and so any geometric description of the structure is required to accommodate  this  in its formulation.  These have consequences in the description of the spacetime geometry and in particular the patching conditions of various tensor like objects  that appear on  spacetime.   Despite
pointing to  a new way of investigating null G-structures, the null structures that appear in supersymmetric backgrounds are special and they must be further generalized to apply to non-supersymmetric solutions.

One of the objectives of this paper is to give the definition of a null G-structure and develop the tools to investigate
the geometry of associated Lorentzian manifolds. In particular, the structure of the tangent and co-tangent bundles of manifolds with a null G-structure will be examined in detail. The emphasis will be on the construction of the fundamental forms of a null G-structure and the exploration of their properties. This will generalize the null structures beyond  the Robinson structure and those that
appear in the context of supersymmetric solutions.   One of the expectations is that  the null G-structures developed here will be sufficiently general to describe the geometry of most of the solutions in four and higher dimensions that have appeared in recent years in the context of string and M-theories.

Under certain conditions, a spacetime with a null G-structure admits a null geodesic congruence generated by a vector field $X$.
This is always the case for spacetimes with a Robinson structure. This will be explored further in the context of null G-structures. The geometry induced from the spacetime with a null G-structure on a null hypersurface ${\cal H}$ transversal to this null geodesic congruence as well as that on the space of orbits ${\cal M}$ of the null geodesics will be determined.  These have applications in the context of black holes and asymptotically flat spacetimes as
 such hypersurfaces are identified with
horizons and    the asymptotic null infinity, respectively.  We shall also give the conditions how to construct a Lorentzian manifold
with null G-structure from data on ${\cal M}$.

Furthermore, we shall examine the local symmetries of null G-structures.  In particular, the symmetry Lie algebras of  a variety of null G-structures will be computed.  These are typically infinite dimensional and  are closely related to the Lie algebras of path groups. Moreover, we shall demonstrate for a certain class of spacetimes the symmetry Lie algebras of some null G-structures interpolate between
the Lorentz symmetry of a Killing horizon hypersurface and the Bondi-Metzner-Sachs (BMS)  algebra \cite{bondi, sachs} of asymptotic null infinity.

In addition to developing the general theory of the null G-structures, we shall investigate in detail the null G-structures  associated with the groups
$SO(n-2)$, $U(k)$ and $SU(k)$, $n=2k+2$,  on a n-dimensional spacetime $M$.  We shall demonstrate that the null structure associated to  $U(k)$  can be identified with that of almost Robinson manifolds. We also give some examples of Lorentzian spacetimes
with a null G-structure which include black holes in all dimensions and brane solutions.

This paper is organized as follows. In section two, we give the definition of a null G-structure and investigate the properties
of the tangent and contangent bundles of the associated spacetime. In section three, we explore the properties of the fundamental forms
of general null G-structures.  We also examine  the induced geometry on null hypersurfaces transversal to null geodesic congruences
generated by the structure as well as that on orbit spaces of null geodesics.  In addition, we  give the definition of invariance of a null G-structure and
present examples of spacetimes with a null G-structure.  In section four, we investigate in detail the geometry of null G-structures
based on the groups $U(k)$ and $SU(k)$ and clarify the relation of the former to the Robinson structure.  In section five, we
investigate the symmetries of null G-structures with emphasis on those related to the groups $SO(n-2)$ and $U(k)$. In particular, we
explain how on certain spacetimes the Lie algebra of spacetime diffeomorphisms which preserves a null G-structure interpolates
between the Lorentz  algebra of a Killing horizon and the BMS  algebra of asymptotic null infinity, and in section six,
we present our conclusions.

\section{Null  $G$-structures}

\subsection{Definition of null $G$-structures}

As it has been mentioned in the introduction, null G-structures are characterised by the existence of a nowhere vanishing
 vector field, considered up to a conformal rescaling,  which defines a null direction on the spacetime. Because of this let us consider
the isotropy group $H_L\subset SO^+(n-1,1)$  of a null line $L$ passing through the origin  in $\bR^{n-1,1}$.  This is spanned by the matrices
\bea
H_L=\left\{ \begin{pmatrix}\ell & -{1\over2}\ell^{-1} v^2 & v^t \\
                  0& \ell^{-1}& {\bf 0}^t\\
                 {\bf 0}& -{1\over\ell} A v& A\end{pmatrix}\vert v\in \bR^{n-2}, A\in SO(n-2), \ell\in \bR-\{0\}\right\}~,
\label{isoline}
\eea
while the isotropy group $H_X$ of  a non-vanishing null vector $X$ is
\bea
H_X=\{(\ell, v, A)\in H_L\vert \ell=1\}~.
\eea
Let us denote with $H_L^+$ the subgroup of $H_L$ for which $\ell>0$. $H^+_L$ is the isotropy group of an {\it oriented} null line and it is more suitable to model a null direction on the spacetime than $H_L$.  Because of this, we shall focus on the use of  $H^+_L$ but most of the analysis that follows with some modifications also applies to $H_L$. Clearly $H_X\subset H^+_L$.

\begin{definition}
A Lorentzian manifold $(M,g)$ admits a (time oriented) null structure iff the structure group of $M$ reduces to a subgroup of $H^+_L$.
\end{definition}

  Clearly a special case of a null structure arises whenever the structure group of $(M,g)$ reduces to a subgroup of $H_X$ instead. In such a case, the null structure is associated to a nowhere vanishing null vector field on the spacetime instead of a null direction.  As the $H_X\subset H^+_L$, manifolds with structure group\footnote{For all manifolds with a $H^+_L$-structure, the topological structure group
reduces to a subgroup of $SO(n-2)$.} $H^+_L$ are more general  than those with structure group $H_X$ so we shall investigate the properties of the former and specialize where necessary on the latter. Another reason to focus on spacetimes with a $H^+_L$-structure
is that  $H^+_L$  accommodates  better the conformal properties of the null structure than $H^+_X$.

Note that $SO(n-2)=\{ (\ell, v, A)\in H_L^+\vert \ell=1, v={\bf 0}\}$ is a subgroup of $H_L^+$. Rewriting $H_L^+$ as $H_L^+(SO(n-2))$,  a distinguished class of subgroups of $H_L^+(SO(n-2))$ are those for which the matrices $A$ are restricted to lie into a subgroup $K$ of $SO(n-2)$.  Denoting these subgroups with $H^+_L(K)$,  the investigation that follows will focus on Lorentzian
manifolds whose structure group is $H^+_L(K)$ characterized by the existence of fundamental forms. The examples that will be explored
in more detail are those null G-structures for which  $K=SO(n-2),   U(k)$ and $SU(k)$, $n=2k+2$.  Though  those with $K=Spin(7)$,  $G_2$ and others are known to occur in certain supersymmetric backgrounds, see \cite{isogroups} for a list of isotropy groups of spinors in $Spin(n-1,1)$ for $n=10, 11$. Another class of subgroups of $H_L^+(SO(n-2))$  that are known to occur as structure groups of Lorentzian manifolds, which again will not be investigated here, are those which are products $H^+_L(K_1)\times K_2$, where $K_1\times K_2\subset SO(n-2)$.  For example compactification backgrounds have such a structure group.

We conclude this section with the definition of the group
\bea
{\mathring H}(K)=\left\{ \begin{pmatrix}\ell &  v^t \\
                 {\bf 0}& A\end{pmatrix}\vert v\in \bR^{n-2}, A\in K\subseteq SO(n-2), \ell\in \bR_{>0}\right\}~.
\label{isolinex}
\eea
As we shall demonstrate this is the G-structure group  on the orbit spaces of null geodesics in Lorentzian manifold with a null $H^+_L(K)$-structure.

\subsection{Geometry of  null structures}

\subsubsection{Vectors, forms and tensors}

Suppose now that $(M,g)$ has a null $H_L^+$ structure. A consequence of this is that the tangent bundle, $TM$, of $M$ admits a  rank 1 subbundle  $N$ with a null fibre.  This is topologically trivial as it  admits a nowhere vanishing section $X$.    As $N$ is a subbundle of $TM$,  we have that
\bea
0\rightarrow N\rightarrow TM\rightarrow TM/N\rightarrow 0~.
\label{nullseq}
\eea
Furthermore define the 1-form  $\kappa$, $\kappa(W)=g(W,X)$, where $W$ is any vector field on $M$. $\kappa$ is also nowhere vanishing and defines a line subbundle $\tilde N$ in $T^*M$.  As a result one also has
\bea
0\rightarrow \tilde N\rightarrow T^*M\rightarrow T^*M/\tilde N\rightarrow 0~.
\label{dualnullseq}
\eea
Notice that the spacetime metric $g$ does not restrict well on either $TM/N$ or $T^*M/\tilde N$ as it depends on the choice
of representative vector in the fibre of these bundles. Observe also  that $\tilde N$ is not the dual of $N$.

Before we proceed further let us emphasize that $X$ and $\kappa$ are specified up a conformal transformation by strictly positive\footnote{The functions are taken to be positive to preserve the orientation defined by $X$.} functions on $M$.
This is because such  transformations retain $X$ and $\kappa$ as sections $N$ and $\tilde N$, respectively. In addition the geometric properties under investigation here depend only of the null direction  of  $X$ and $\kappa$ instead of $X$ and $\kappa$ themselves.   These null directions are not affected by such conformal rescaling. Although $\kappa$ and $X$ are related via a metric, we shall take the conformal rescalings of $X$ and $\kappa$ to be independent unless otherwise is explicitly
stated.

Next consider $N^\perp$, the orthogonal subbundle of $N$ in $TM$, whose fibres are
\bea
N_p^\perp=\{v\in T_pM\vert \kappa(v)=0\}.
\eea
Note that two metrics in the same conformal class give rise to the same $N^\perp$.   Clearly $N$ is a subbundle in $N^\perp$ and so
\bea
0\rightarrow N\rightarrow N^\perp\rightarrow N^\perp/N\rightarrow 0~.
\label{nullseqx}
\eea
The spacetime metric $g$ restricts well on the fibres of $N^\perp/N$.  In fact it restricts to a non-degenerate Euclidean signature
metric.  This is significant in many general relativity computations.

For reasons that will become more apparent later, let us develop a tensor calculus based on $TM/N$ and $T^*M/\tilde N$.  Let us denote the smooth sections of a vector bundle $E$ over $M$ with $\Gamma(E)$.
First observe that
 $\hat W\in \Gamma(TM/N)$ and $\hat\alpha\in \Gamma(T^*M/\tilde N)$ can be represented by a vector field $W$ and a  1-forms  $\alpha$ on the spacetime $M$, respectively,
up to the equivalence relation $\sim$, where $W\sim W'$  and $\alpha\sim \alpha'$, iff $W=W'+a X$ and $\alpha=\alpha'+ b \kappa$, for some spacetime functions $a$ and $b$.  This can be extended to sections of the tensor bundle $\otimes^k (TM/N) \otimes^\ell  (T^*M/\tilde N)$.  In particular  $\hat I\in\Gamma\left((TM/N) \otimes  (T^*M/\tilde N)\right)$ can be represented by an (1,1) tensor $I$ up to the equivalence $\sim$, where
$I\sim I'$ iff $I=I'+X\otimes \alpha+ W\otimes \kappa$ with  $\alpha$ and $W$ be  a 1-form and a vector field on the spacetime, respectively.
It is often convenient to use instead of $I\sim I'$, $I=I'~ \mathrm{mod}~(X, \kappa)$.

Similarly sections  of $N^\perp/N$ and $\tilde N^\perp/ \tilde N$ can be viewed as sections of $N^\perp$ and $\tilde N^\perp$ up to the same identifications as those for the sections of of $TM/N$ and $T^*M/\tilde N$ described in the previous paragraph.  However, the identification for the sections of $\otimes^k (N^\perp/N) \otimes^\ell  (\tilde N^\perp/\tilde N)$ is somewhat different from of that in $\otimes^k (TM/N) \otimes^\ell  (T^*M/\tilde N)$. In particular $\check I\in \Gamma\left((N^\perp/N)\otimes (\tilde N^\perp/\tilde N)\right)$ can be viewed as a section $I$ of $N^\perp\otimes \tilde N^\perp$  up to the identification $I=I'+X\otimes \alpha+ W\otimes \kappa$, where now $\alpha\in \Gamma(\tilde N^\perp)$ and $W\in \Gamma( N^\perp)$ instead of sections of $T^*M$ and $TM$, respectively.

Next define a Lie derivative type of operation, ${\hat{\cal L}}_W$,  with respect to a vector field $W$ on $M$ on $\Gamma(TM/N)$   as follows
\bea
{\hat{\cal L}}_W \hat V\defeq [W, V]~ \mathrm{mod}~X~,~~~W\in \Gamma(TM), ~~~\hat V\in \Gamma(TM/N)~.
\label{liedef}
\eea
This operation is not always well defined unless $W$ is appropriately restricted.  In particular,  one has the following.
\begin{prop}
${\hat{\cal L}}$ is well defined provided that $W$ preserve the null structure associated to $X$, i.e. $[W, X] ~\mathrm{mod}~X=0$
\end{prop}
\begin{proof}
Indeed consider another representative, $ V+f X$, of the section $\hat V$ of $TM/N$.  Then
\bea
{\hat{\cal L}}_W  ( \reallywidehat{V+f X})=[W,  V+f X]~\mathrm{mod}~X={\hat{\cal L}}_W  \hat V+ f [W, X]~\mathrm{mod}~X= {\hat{\cal L}}_W  \hat V~,
\eea
where in the last equality we have used the assumption that $W$ preserves the null structure.
\end{proof}

Observe if $W$ and $W'$ preserve the null structure, $[W,W']$ also preserves the null structure.  Using this, one can prove that
\bea
{\hat{\cal L}}_{[W, W']}\hat V={\hat{\cal L}}_W {\hat{\cal L}}_{W'} \hat V-{\hat{\cal L}}_{W'} {\hat{\cal L}}_{W}\hat V~.
\eea
Similarly a Lie derivative can be defined on  $\Gamma(T^*M/\tilde N)$ as
\bea
{\hat{\cal L}}_W \hat\alpha\defeq {\cal L}_W \alpha~\mathrm{mod}~\kappa~.
\eea
This is well defined provided that again $W$ preserves the null structure associated to $\kappa$, i.e. ${\cal L}_W\kappa~\mathrm{mod}~\kappa=0$. This can be generalized to sections of $\otimes^k(TM/N)\otimes^\ell (T^*M/\tilde N)$ by applying appropriately the $\mathrm{mod}~(X, \kappa)$ operation.
In such a case ${\hat{\cal L}}_W$ is well-defined provided that $W$ preserve both $X$ and $\kappa$ null structures.  Observe  that  the condition  ${\cal L}_W\kappa~\mathrm{mod}~\kappa=0$  is not implied from the ${\cal L}_W X~\mathrm{mod}~X=0$  unless $W$ is further restricted to satisfy $i_X {\cal L}_W g= 0 ~\mathrm{mod}~\kappa~$.
Furthermore an exterior derivative can be defined on the sections of exterior bundle  $\Lambda^k (TM/N)$ of $TM/N$ as $\hat d \hat\alpha= d\alpha ~\mathrm{mod}~\kappa$. This is well defined provided that $d\kappa=0~\mathrm{mod}~\kappa$ (or equivalently $\kappa\wedge d\kappa=0$).

A Lie derivative operation $\check {\cal L}$ can also be defined on  $\Gamma(\otimes^k(N^\perp/N)\otimes^\ell (\tilde N^\perp/\tilde N))$.  However now the coefficient tensors that appear in the modulus operation are sections of either
$\otimes^{k-1}N^\perp\otimes^\ell \tilde N^\perp$ or $\otimes^{k}N^\perp\otimes^{\ell-1} \tilde N^\perp$.  For example let $\check \alpha\in \Gamma(\Lambda^2(\tilde N^\perp/\tilde N))$, then one has $\check{\cal L}_W\check\alpha= {\cal L}_W\alpha ~ \mathrm{mod}~\kappa$,
where the modulus operation is up to sections $\beta\wedge \kappa$ with $\beta\in\Gamma(\tilde N^\perp)$.  The operation
of $\check {\cal L}_W$ is well defined provided that $W$ preserve
the null structure associated to both $X$ and $\kappa$ in all cases.

Although this will not be used later for completeness consider  a connection $D$ in $TM$ and define a connection $\hat D$ in $TM/N$ as
\bea
\hat D_W \hat V\defeq D_W  V ~\mathrm{mod}~ X~,~~~W\in \Gamma(TM)~.
\label{condef}
\eea
It can be seen that $\hat D$ is well defined provided that $D$ satisfies
\bea
D_W X=\eta(W) X~,
\label{dwx}
\eea
where $\eta$ is a spacetime 1-form.   Similarly given a connection $D$ in $T^*M$, one can define a connection
$\hat D$ in $T^*M/\tilde N$ as $\hat D_W \hat\alpha\defeq  D_W \hat\alpha ~\mathrm{mod}~ \kappa$ provided that
$D_W \kappa=\theta(W) \kappa$ for some spacetime form $\theta$.  One can then extend $\hat D_W$ to the sections of
$\otimes^k(TM/N)\otimes^\ell (T^*M/\tilde N)$.

\subsubsection{Splitting of the tangent and cotangent bundles}

There is not a natural way to identify $TM/N$ as a subbundle of $TM$.  However, this will be the case
if one  splits of the sequence (\ref{nullseq}).  One way to do this is to choose\footnote{Note that $\lambda$ could also be chosen as $\lambda(X)=f$, where $f$ is a no-where vanishing function of $M$.  However, in such a case one could consider $f^{-1}\lambda$ instead.} a null 1-form  $\lambda$ on $M$, $g^{-1}(\lambda, \lambda)=0$,  such that
$\lambda(X)=1$.  Then $TM=N\oplus Z$,  where $Z$ is the bundle whose fibres, $Z_p$, are $Z_p=\{v\in T_pM\vert \lambda(v)=0\}$.  Indeed notice that the splitting map $\hat\lambda:~TM/N\rightarrow TM$ defined as $\hat\lambda(\hat W)\defeq W-\lambda(W) X$ is independent from the
representative $W$ of $\hat W$, is an injection and the image of $\hat\lambda$ is $Z$.

Similarly, a splitting can be chosen for the sequence (\ref{nullseqx})  via the use of the 1-form $\lambda$ as above now restricted on the sections of $N^\perp$.   As a result
 $N^\perp=N\oplus T $, where the fibres of $T$, $T_p$,  are $T_p=\{u\in N_p^\perp\vert \lambda(u)=0\}$.  $T$ is the image of the
 splitting map $\hat\lambda$ now restricted on $N^\perp/N$.
 The metric of $M$ can be decomposed as
 \bea
 g=\kappa\otimes \lambda+\lambda\otimes \kappa+ g_T~,
 \eea
 where  $g_T$ is the restriction of $g$ on the fibres of $T$. $g_T$ is a metric with Euclidean signature on the fibres of $T$. $T$ is also called a screening space.

One can also define a splitting of the sequence (\ref{dualnullseq}) for the cotangent bundle. In particular  using
$Y(\alpha)\defeq g^{-1} (\alpha, \lambda)$, one has that  $T^*M=\tilde N\oplus \tilde Z$, where $\tilde Z_p=\{\alpha \in T_p^*M\vert Y(\alpha)=0\}$.
Similarly, one defines the orthogonal bundle to $\tilde N$,  $\tilde N^\perp$, in $T^*M$ and the decomposition $\tilde N^\perp=\tilde N\oplus \tilde T$, where $\tilde T_p=\{\alpha\in   \tilde N_p^\perp\vert Y(\alpha)=0\}$.

\begin{definition}
Let $(M,g)$ be a Lorentzian manifold with a null structure contained in $H^+_L(SO(n-2))$. A k-form $\alpha$ on  $(M,g)$ is null iff $\kappa\wedge \alpha=0$.
\end{definition}

\begin{prop}
Let $\alpha$ be a null k-form, then $\alpha=\kappa\wedge \beta$, where $\beta$ is a (k-1)-form.
\end{prop}
\begin{proof}
Indeed take the inner derivation of $\kappa\wedge \alpha$ with respect to $Y$ to find
\bea
i_Y(\kappa\wedge \alpha)=\alpha-\kappa\wedge i_Y \alpha=0~,
\eea
where we have used that $i_Y(\kappa)=\kappa(Y)=\lambda(X)=1$. This proves the statement.
\end{proof}
Notice that in general $\beta$ depends on the choice of splitting. However its class $\hat\beta$ as a section of $\Lambda^{k-1}(TM/N)$
does not.

A consequence of the above lemma is that if $\kappa\wedge\alpha=\kappa\wedge \beta$ for some k-forms $\alpha$ and $\beta$, then $\alpha=\beta+\kappa\wedge \gamma$ for some (k-1)-form $\gamma$. Note also that the operation $\delta_\kappa(\alpha)\defeq \kappa\wedge \alpha$ is a cohomology operation on the space of forms as $(\delta_\kappa)^2=0$.  The above proposition implies that
all the cohomology groups of $\delta_\kappa$ are trivial.

\subsubsection{Dependence  on  the choice of splitting}\label{choicesplit}

 The subbundle  $T$ of $TM$ depends on the choice of splitting $\lambda$ and we  denote this with $T_\lambda$.   To investigate the dependence of $T_\lambda$ on the choice of $\lambda$ suppose that $\lambda'$ is another choice of splitting with $\lambda'(X)=1$.  Any 1-form can be written as $\lambda'= a \lambda+ b \kappa+  \gamma$,  where $\gamma\in\Gamma(\tilde T_\lambda)$.  Imposing $\lambda'(X)=1$, one finds that
$\lambda'=\lambda+ b \kappa + \gamma$.  Furthermore imposing the condition that $\lambda'$ must be null, one finds that
\bea
\lambda'=\lambda-{1\over2} \parallel \gamma\parallel^2 \kappa+ \gamma~.
\eea
 To compare $T_{\lambda'}$ with $T_\lambda$ consider a section $Z_{\lambda'}$ of $T_{\lambda'}$.
Writing $Z_{\lambda'}= a X+ bY+ Z_{\lambda}$, where $Z_{\lambda}\in \Gamma(T_{\lambda})$, as any vector field on $M$ can be decomposed in this way, and imposing the conditions $\kappa(Z_{\lambda'})=\lambda'(Z_{\lambda'})=0$, one finds that
\bea
Z_{\lambda'}=Z_{\lambda}-\gamma(Z_{\lambda}) X~.
\eea
Therefore $T_{\lambda'}$ is ``shifted'' relative to $T_{\lambda}$ with $N$ as expected.  Notice that $g_{T_{\lambda'}}=g_{T_{\lambda}}$ because as it has already been mentioned  $g$ restricts well as a Euclidean signature fibre metric on $N^\perp/N$.

Next let us determine $Y_{\lambda'}$ in terms of $Y_\lambda$,  where again the subscripts denote the dependence of the vector fields
on the choice of splitting.  Writing again $Y_{\lambda'}=a Y_\lambda + bX+ Z$,  where $Z\in\Gamma(T_\lambda)$, and imposing
that $\kappa(Y_{\lambda'})=1$, $\lambda'(Y_{\lambda'})=0$ and that $g(Y_{\lambda'},Z_{\lambda'})=0$,  where $Z_{\lambda'}$ is any section of
$T_{\lambda'}$, one finds that
\bea
Y_{\lambda'}=Y_\lambda-{1\over2}  \parallel\gamma\parallel^2 X+ W~,
\eea
where $W(\alpha)= g_T^{-1} (\gamma, \alpha)$ for every  $\alpha\in \Gamma(\tilde T_\lambda)$.

It remains to find the way that the subbundle $\tilde T$ of $T^*M$ depends on the splitting.  For this consider a section $\alpha_{\lambda'}$ of $\tilde T_{\lambda'}$.  Then  $\alpha_{\lambda'}= a \kappa+b \lambda+ \alpha_\lambda$, where $\alpha_\lambda$ is a section of $\tilde T_\lambda$. Imposing the
conditions $i_X \alpha_{\lambda'}=i_{Y_{\lambda'}}\alpha_{\lambda'}=0$ on $\alpha_{\lambda'}$, so that $\alpha_{\lambda'}$ is a section of $\tilde T_{\lambda'}$, one finds that
\bea
\alpha_{\lambda'}=-\alpha_\lambda(W) \kappa+ \alpha_\lambda~.
\eea
Therefore $\tilde T_{\lambda'}$ is ``shifted'' relative to $\tilde T_{\lambda}$ with the line bundle $\tilde N$.

\subsection{Null geodesic congruences and null structures}\label{ngc}

 Null structures are closely related to the existence of null geodesic congruences in a spacetime.  For this impose the condition
  \bea
  \kappa\wedge {\cal L}_X \kappa=0~,
  \label{kappainv}
  \eea
  on $\kappa$.
 A key consequence of (\ref{kappainv}) is as follows \cite{robtr}.
 \begin{prop}\label{nullgeo}
 If $X$ and $\kappa$ satisfy (\ref{kappainv}), then integral curves of $X$ will be null geodesics.
 \end{prop}
 \begin{proof}
  Indeed
 \bea
 \kappa\wedge \nabla_X \kappa=\kappa\wedge i_X d\kappa=\kappa\wedge {\cal L}_X \kappa=0~,
 \eea
 where we have used that $\kappa(X)=0$. So $\nabla_X X\parallel X$.
 \end{proof}

 This is significant as the boundaries\footnote{The use of the term ``boundary''  does not signify that the spacetime necessarily ends at those hypersurfaces. Instead it is used to denote
   hypersurfaces of interest like the null  infinity for asymptotically flat spacetimes or the event horizon of a black hole.} of spacetimes are  orbit spaces
 of null geodesic congruences, e.g. the even horizons of black holes as well as the conformal boundaries at infinity of asymptotically flat spacetimes. This will be explored further below.

 Clearly  (\ref{kappainv}) holds whenever $X$ is Killing, ${\cal  L}_Xg=0$,
as ${\cal  L}_X \kappa=0$. So from the proposition \ref{nullgeo}, one has  $\nabla_X X=0$ and the integral curves of $X$ are null geodesics. All supersymmetric backgrounds whose Killing spinor bilinears include  a null Killing vector belong to this special case.

 Another special case that arises is whenever $\kappa\wedge d\kappa=0$.  In such a case, the spacetime is foliated with leaves $(n-1)$-dimensional hypersurfaces. Taking the inner derivation of $\kappa\wedge d\kappa=0$ with $X$, one arrives as (\ref{kappainv}).
 So again the integral curves of $X$ are null geodesics.

\subsection{Frames and null structures}

In the description of geometry of a spacetime $M$ with a null G-structure with a splitting, it is often useful to introduce a local co-frame $\{\fe^-, \fe^+, \fe^i; \,i=1, \dots, n-2\}$ such that
\bea
\fe^-=\kappa~,~~~\fe^+=\lambda~.
\eea
Such  a notation has also extensively been  used in the investigation of supersymmetric backgrounds.

The metric is written in terms of the co-frame as $g=2 \fe^-\fe^++\delta_{ij}\fe^i\fe^j$.  Under local transformations in the isotropy group $H^+_L$, the co-frame transforms as
\bea
&&\fe^-\rightarrow \ell^{-1} \fe^-~,~~~\fe^+\rightarrow \ell \fe^+-{1\over 2} \ell^{-1} v^2 \fe^-+ v_k\fe^k~,
\cr
&&\fe^i\rightarrow A^i{}_j (\fe^j- \ell^{-1}  v^j \fe^-)~.
\label{frametrans}
\eea
Clearly the transformation of the co-frame under the isotropy group $H_X$ is as above but now $\ell=1$.  Observe that
a change of splitting, investigated in section \ref{choicesplit},  introduces a transformation on the frame of the type described above and so the spacetime metric
does not depend on the choice of splitting as expected.

A local description of the geometry is as follows.  Adapting a coordinate along $X$, $X=\partial_u$,  introducing additional coordinates $v, y^I$ and choosing a splitting, one can write for the coframe
\bea
\fe^-=  h (dv+n_I dy^I)~,~~~  \fe^+=du+{1\over2} V dv+m_I dy^I~,~~~\fe^i= e^i_I  dy^I ~,
\label{locframe}
\eea
where $h$, $V$, $n_I$, $m_J$ and $e^i_J$ depend on all coordinates. The frame then is
\bea
\fe_+= \partial_u~,~~\fe_-=h^{-1} (\partial_v-{1\over2} V\partial_u)~,~~\fe_i=e_i^I (\partial_I-  n_I\partial_v-m_I \partial_u+{1\over2}  V n_I \partial_u)~.
\eea
So one has $X=\fe_+$ and $Y=\fe_-$.
The above local choices of frame and co-frame are not unique.  In particular the $\fe^i$ frame is chosen as $i_Y \fe^i=0$.  This choice
is helpful in the local description of fundamental forms of a null G-structure as it will be explained in section  \ref{sec:form}.  However,
one can also choose $\fe'^i=\fe^i+p^i \fe^-$.  The form of the spacetime metric does not change as the addition of $p^i \fe^-$ can be compensated in a redefinition of $V$ and $m$.

\section{General null G-structures}

\subsection{Forms and null G-structures}\label{sec:form}

Before we turn to investigate examples of individual  null G-structures, it is instructive to give an overview of some of their  properties. Suppose a  spacetime $M$  admits a  $H^+_L(K)$-structure characterized with fundamental forms
  $\kappa$ and one or more additional forms that we collectively denote with   $\chi$.  All the fundamental forms of such  a null structure   satisfy the conditions
 \bea
 \kappa\wedge \chi=0~, ~~~i_X \chi=0~.
 \label{chicon}
 \eea
  These are forms on the spacetime and so  do not
depend on the choice of spitting $\lambda$ of $TM$.  Moreover all the fundamental forms are specified up to a conformal rescaling by strictly positive  spacetime functions. This is because the properties (\ref{chicon}) are independent from the overall normalization of these forms. However in the presence of more than one fundamental forms $\chi$ that satisfy certain relative normalisation conditions not all such rescaling can be independent. An example of this arises in the investigation of the null $H^+_L(SU(k))$-structures below.  However note that not all  fundamental forms of a $H^+_L(K_1)\times K_2$-structure are null as the fundamental forms associated to $K_2$ do not satisfy the first condition in (\ref{chicon}).

The first condition in (\ref{chicon}) implies  that $\chi$  is null and it can be solved to yield $\chi=\kappa\wedge \alpha$ for some form $\alpha$. Indeed taking the inner derivation of the first equation in (\ref{chicon}) with $Y$ and using $\kappa(Y)=1$ yields the desirable result.  Clearly $\alpha$ is not unique as both $\alpha$ and $\alpha+\kappa\wedge \gamma$, for any form $\gamma$, give rise to the same $\chi$.  Typically,  $\alpha$ depends on the choice of splitting
 of $TM$ but not its class $\hat\alpha$ as a section of $\Lambda^*(TM/N)$.

 \begin{prop}\label{splitprop}
 Given a splitting of $TM$,  a null fundamental k-form $\chi$ of a null G-structure can be represented as $\chi=\kappa\wedge \phi$, where $\phi$ is a section of $\Lambda^k(T)$, i.e. it satisfies  $i_X\phi=i_Y\phi=0$.
 \end{prop}
\begin{proof}
First let us demonstrate $\chi=\kappa\wedge \alpha$,  where $i_Y\alpha=0$.  Indeed as $\chi$ can be written as $\chi=\kappa\wedge \beta$ for some form $\beta$, we set
 $\alpha=\beta-\kappa\wedge i_Y \beta$.    Then clearly $\chi=\kappa\wedge \beta=\kappa\wedge \alpha$ and $i_Y \alpha=0$.

 To continue decompose $\alpha$ as $\alpha=\lambda\wedge \zeta+ \phi$, where $\zeta$ and $\phi$ are sections of $\Lambda^{k-1}(T)$ and $\Lambda^k(T)$, respectively.  For this we have used that $i_Y\alpha=0$.  Therefore, we have that
 $\chi=\kappa\wedge \lambda\wedge \zeta+ \kappa\wedge \phi$.  Imposing now the condition $i_X\chi=0$ in (\ref{chicon}), we find that
 $\kappa\wedge \zeta=0$ and so $\chi=\kappa\wedge \phi$, where $i_X\phi=i_Y\phi=0$.
\end{proof}

For a  null G-structure $H^+_L(K)$  the representative $\phi$ of $\chi$, as described in the previous proposition,
is identified with a  fundamental form of the compact subgroup $K$. However other representatives have also been used in the literature to describe a null G-structure on a spacetime. Of course  it is  required  that the geometry of spacetime to be independent of these choices.

It is clear from the comment above that the fundamental forms of a null $H^+_L(K)$-structure can be constructed from those of $K$.
In particular, the $H^+_L(SO(n-2))$-structure apart from $\kappa$ also admits a fundamental $(n-1)$-form $\chi$ which can be identified
with the fibre volume form of $N^\perp$, $\chi=d\mathrm{vol}(N^\perp)$.  A similar construction will be made below for the fundamental
forms of the $H^+_L(U(k))$ and $H^+_L(SU(k))$ structures.

In an adapted  frame to a null G-structure associated to the forms $\kappa$ and $\chi$ as in (\ref{locframe}), one has $\kappa=\fe^-$ and $\chi=\fe^-\wedge \phi$, where $\phi={1\over (k-1)!} \phi_{i_1\dots i_{k-1}} \fe^{i_1}\wedge \cdots\wedge \fe^{i_{k-1}}$.  The components of $\phi$ in the frame are constant.  Observe that $\phi$ is not covariant under the patching conditions (\ref{frametrans}) and depends on the choice of splitting.  Though of course $\chi$ does not.

\subsection{Geometry of null transversal hypersurfaces}

Many of the spacetime boundaries, like Killing horizons and  conformal boundaries of asymptotically flat spacetimes, are null hypersurfaces
which are transversal to null geodesic congruences.

\begin{definition}
A hypersurface ${\cal H}$ in $M$ is transversal to the flow of a vector field $X$ iff $X$ is nowhere tangent to ${\cal H}$.
\end{definition}

Null geodesic congruences have many transversal hypersurfaces, the focus here will be on transversal hypersurfaces which in addition  are {\it null},
i.e. hypersurfaces for which their tangent space, $T_p{\cal H}$, at every point $p\in {\cal H}$, is a null subspace of $T_pM$. Such hypersurfaces admit  normal vector field ${\bf n}$ which is null and so simultaneously tangent to the hypersurface. In such a case, $g(X, {\bf n})\not=0$  as otherwise $X$, which generates the null geodesic congruence,  will
be tangent to ${\cal H}$. So after a possible rescaling of ${\bf n}$, one can set $g(X, {\bf n})=1$.

As ${\bf n}$ is nowhere vanishing it defines a trivial line bundle ${\cal N}$ in $T{\cal H}$ and so one has
\bea
0\rightarrow {\cal N}\rightarrow T{\cal H}\rightarrow T{\cal H}/{\cal N}\rightarrow 0
\eea
This sequence can be split using the pull back  $j^*\kappa$ of $\kappa$ on ${\cal H}$, where $j$ is the inclusion of ${\cal H}$ in $M$,  as $j^*\kappa({\bf n})=g(X, {\bf n})=1$. So one has $T{\cal H}={\cal N}\oplus T^{\cal H}$,  where $T^{\cal H}_p=\{ v\in T_p{\cal H}\vert j^*\kappa(v)=0\}$.  Observe
that the spacetime metric restricted on ${\cal H}$ becomes a positive definite fibre metric on  $T^{\cal H}$.
Similarly
$T^*{\cal H}=\tilde {\cal N}\oplus \tilde T^{\cal H}$,  where now the fibres of $\tilde {\cal N}$ are spanned by $j^*\kappa$ and
$\tilde T^{\cal H}_p=\{\alpha\in T_p^*M\vert \alpha({\bf n})=0\}$.  Note that the above two described splittings are natural in the sense
that both $j^*\kappa$ and ${\bf n}$ are determined in terms of  the null structure and the choice of hypersurface. No other
arbitrary choices are involved.

Incidentally, it is well known that ${\bf n}$ generates a null geodesic congruence in ${\cal H}$.  Indeed
$\nabla_{\bf n} g({\bf n}, {\bf n})= 2 g(\nabla_{\bf n}{\bf n}, {\bf n})=0$ implies that $\nabla_{\bf n}{\bf n}$ is tangent
to ${\cal H}$.  Then for any other tangent vector field $Z$ to ${\cal H}$, $\nabla_{\bf n} g({\bf n}, Z) =g(\nabla_{\bf n}{\bf n}, Z)+g({\bf n}, \nabla_{\bf n}Z)=g(\nabla_{\bf n}{\bf n}, Z)=0$, where $g({\bf n}, \nabla_{\bf n}Z)$ vanishes as a consequence
of the torsion free condition of $\nabla$ and that $[{\bf n}, Z]$ is tangent to ${\cal H}$.  Therefore $\nabla_{\bf n}{\bf n}\parallel {\bf n}$.

\begin{theorem}
Suppose that $M$ admits a null $H^+_L(K)$-structure admitting fundamental forms $\kappa$ and $\chi$, and ${\cal H}$ be a null hypersurface in $M$ transversal to a null geodesic congruence generated by $X$.  Then ${\cal H}$ admits a $\bR^+\times K$-structure, where $\bR^+\times K$ is the isotropy group of the conformal class of ${\bf n}$ in $H^+_L(K)$.
\end{theorem}
\begin{proof}
Recall that the form $\chi$ satisfies, $\kappa\wedge \chi=i_X\chi=0$. As $j^*(\kappa\wedge \chi)=j^*\kappa\wedge j^* \chi=0$, $j^*\chi=j^*\kappa\wedge \omega$, where $j$ is the inclusion of ${\cal H}$ in $M$ and $\omega$ is a section of $\Lambda^{k-1}T^{\cal H}$. Indeed acting on $j^*\kappa\wedge j^* \chi=0$ with $i_{\bf n}$, one finds that $\omega=i_{{\bf n}} j^*\chi$ and so $i_{{\bf n}}\omega=0$.      The forms $j^*\kappa$ and  $\omega$ are the fundamental forms of $\bR^+\times K$ and so  ${\cal H}$ admits $\bR^+\times K$-structure.  Note that if $H^+_L(K)$ admits more than one fundamental forms $\chi$ that satisfy non-trivial normalization conditions, these are also automatically satisfied on ${\cal H}$.
\end{proof}

One of the consequences of the above proposition is as follows.

\begin{corollary}
 Let $M$ admit a $H^+_L(SO(n-2))$  structure and ${\cal H}$ be a null hypersurface in $M$ defined as in the proposition above.   Then ${\cal H}$ is oriented.
\end{corollary}

\begin{proof}
As $M$ admits a null nowhere vanishing (n-1)-form $\chi$.  The pull-back of $\chi$ on ${\cal H}$ defines a nowhere vanishing top
form on ${\cal H}$ and so the hypersurface is oriented.
\end{proof}

\subsection{Orbit spaces of null geodesics}\label{sec:gorbit}
\subsubsection{Geometry of null geodesic orbit spaces}

In many examples of interest, the geometry of a   spacetime with a $H^+_L$ structure can be described as a fibration. To  investigate the conditions required for this, let us assume that there is an open set $U\subseteq M$ such that the foliation on $M$ generated by the flow of $X$ is regular, i.e. the orbit space ${\cal M}$ of the integral curves of $X$ in $U$ is a manifold and the projection of $p:\,U\rightarrow {\cal M}$ is a surjection. $U$ may be thought of as a neighbourhood of a spacetime boundary.

To find whether a null structure of a spacetime can be projected on ${\cal M}$ consider the following definition whose justification
is provided in the theorem that follows.
 \begin{definition}
 A fundamental  form $\chi$ of a $H^+_L(K)$-structure is preserved  under the  flow generated by $X$ iff
\bea
{\cal L}_X \chi= b\, \chi ~,
\label{flowinv}
\eea
for some spacetime function $b$ which may depend on $\chi$.
\end{definition}
It should be noted that the definition above does not depend on the choice of $X$ in its conformal class. Any other choice will
lead to (\ref{flowinv}) up to an appropriate redefinition of the function $b$.  The same applies for the choice of $\chi$ in its conformal class.

\begin{theorem}\label{th:orbit}
Suppose that the spacetime $M$ admits a  $H^+_L(K)$-structure with fundamental forms $\kappa$ and $\chi$,  and that both these forms  are preserved under the
 flow generated by $X$.  In addition assume that  there is a open set $U\subseteq M$ such that the foliation generated by $X$ has orbit space ${\cal M}$ and $U$ admits a hypersurface ${\cal S}$ which is nowhere tangent to $X$.   Then ${\cal M}$  admits a  null ${\mathring H}(K)$-structure represented by forms in the conformal class of  $\kappa$ and $\chi$.
\end{theorem}

\begin{proof}
First notice that the invariance  on $\kappa$, ${\cal L}_X\kappa= a \kappa$, under  the flow generated by $X$,  is equivalent to $\kappa\wedge {\cal L}_X \kappa=0$.  So from the results of section \ref{ngc}, $X$ generates
a null geodesic congruence.  Consider the projection $\pi:  U\rightarrow {\cal M}$.  The necessary and sufficient conditions
for a form $\alpha$ on $U$ to arise as a pull-back of a form $\beta$ on ${\cal M}$, $\alpha=\pi^*\beta$, are $i_X \alpha={\cal L}_X\alpha=0$.

Clearly $\kappa$ and $\chi$ satisfy the first condition, see (\ref{chicon}),  but not the Lie derivative condition.  However both $\kappa$ and $\chi$
can be considered up to a conformal factor.  So demanding invariance of  $e^f \kappa$ under the flow of $X$, for some function  $f$  on $U$,
gives the differential condition
\bea
X(f)=a~.
\label{defeqn}
\eea
A similar differential equation can be derived for $\chi$.

Differential equations as in (\ref{defeqn}) have been considered before in the context of manifolds \cite{jmlie} and they admit a unique solution.  In particular given a hypersurface ${\cal S}$ in $U$ such that $X$ is nowhere tangent on ${\cal S}$, the (\ref{defeqn}) has a unique solution in a neighbourhood $W\subseteq U$
of $S$ in $U$ such that $f\vert_{\cal S}=q$, where $q$ is a smooth function of ${\cal S}$.

Therefore from the assumptions of the theorem,   the forms $ e^f \kappa$ and $e^h \chi$ exist and they are the pull-back of the forms
$\kappa_{\cal M}$ and $\chi_{\cal M}$ on ${\cal M}$, respectively, where $X(f)=a$ and $X(h)=b$.  Furthermore $\kappa_{\cal M}\wedge \chi_{\cal M}=0$ as $\pi^*$ is an inclusion.
As under the projection $\pi_*X=0$,  the structure group of ${\cal M}$ is ${\mathring H}(K)$, see (\ref{isolinex}).  $T^*{\cal M}$ admits a trivial line bundle ${{\tilde{\cal N}}}$ whose fibres are spanned by $\kappa_{\cal M}$ but it does not naturally split as $T^*{\cal M}={\tilde{\cal N}}\oplus \tilde T^{\cal M}$.  Of course a splitting can always be arranged by choosing a dual vector field to $\kappa_{\cal M}$.

\end{proof}

A priori the solutions of (\ref{defeqn}) and so the choice of forms $ e^f \kappa$ and $e^h \chi$ depends on the choice of the hypersurface ${\cal S}$ and the choice of the boundary condition $f\vert_{\cal S}=q$ on ${\cal S}$.  As two solutions $f_1$ and $f_2$ of (\ref{defeqn}) satisfy
$X(f_1-f_2)=0$ and so $f_1=f_2+p$, where $p$ is a function of ${\cal M}$. The arbitrariness in the choice of hypersurface ${\cal S}$ and boundary condition can be compensated in the choice of the conformal class of $\kappa_{\cal M}$ and $\chi_{\cal M}$ on ${\cal M}$.

It should be noted that if a $H^+_L(K)$-structure has fundamental form $\kappa$ and more than one forms $\chi$ that
are required to obey  relative normalisation conditions, ${\cal M}$ is induced with  ${\mathring H}(K)$-structure  provided that the flow of
$X$ apart from individual forms also preserves their relative normalization conditions.  An example of this arises in the
description  of the $H^+_L(SU(k))$ structure, see section \ref{sec:suk}.

\begin{corollary}
 Let $M$ admit a $H^+_L(SO(n-2))$-structure and $\kappa$ be preserved by the $X$ flow, then $d\mathrm{vol}(N^\perp)$ is preserved by the $X$ flow.  Moreover  if the remaining   conditions of the theorem \ref{th:orbit} are met, then ${\cal M}$ is oriented.
\end{corollary}

\begin{proof}
We have demonstrated that if $M$ admits a $H^+_L(SO(n-2))$ structure, there is a nowhere vanishing (n-1)-form $d\mathrm{vol}(N^\perp)$
on $M$. Now if ${\cal L}_X\kappa=a\, \kappa$, then ${\cal L}_X d\mathrm{vol}(N^\perp)= b\, d\mathrm{vol}(N^\perp)$,  where $b=a+i_X d\fe^i{}_i$.
Therefore, the $H^+_L(SO(n-2))$ structure is preserved by the flow of $X$.  In particular, $ d\mathrm{vol}(N^\perp)$ induces a
nowhere vanishing top form on ${\cal M}$ and so it is oriented.
\end{proof}

Next let us investigate the conditions for the component $g_T$ of the metric $g$ of $M$ to arise from an appropriate tensor on ${\cal M}$.  First observe that the push forward bundle $\pi_* T$ on ${\cal M}$ is independent from the choice of the splitting $\lambda$ used to define $T$.  Because of this, one can define $\pi_*(N^\perp/N)\defeq \pi_* T$. After choosing a splitting,  from construction one has $i_X g_T=0$. To continue suppose there is a splitting such that the  condition ${\cal L}_X g_T= c\, g_T$ holds.
Then using the same arguments as in theorem \ref{th:orbit}, there is a tensor $g_{\cal M}$ such that $\pi^* g_{\cal M}= e^f\, g_T$
for some function $f$ of $M$ with $X(f)=c$.  The tensor $g_{\cal M}$ is a fibre metric on the subbundle $\pi_* (N^\perp/N)$ in $T{\cal M}$.

\subsubsection{Reconstruction of $M$ from ${\cal M}$}

One can reconstruct a Lorentzian manifolds $M$  with a null G-structure from the geometric data on ${\cal M}$ given by $\kappa_{\cal M}$, $\chi_{\cal M}$,
and $g_{\cal M}$ as follows.  Consider $M=\bR\times {\cal M}$.  Set $X=\partial_u$, where $u$ is the coordinate of $\bR$. The fundamental
forms of the null G-structure on $M$ are $\kappa={a_1}\, \pi^*\kappa_{\cal M}$ and $\chi= {a_2}\, \pi^*\chi_{\cal M}$, where $a_1, a_2$
are strictly positive functions on $M$.  Note however that if there are two or more fundamental forms $\chi$, the conformal factors $a_2$ should be chosen
in such a way that they preserve their relative normalization. Next a compatible metric can be introduced on $M$ as
\bea
g=a_1 \pi^*\kappa_{\cal M}\otimes \lambda+a_1 \lambda \otimes \pi^*\kappa_{\cal M}+ a_3 \pi^* g_{\cal M}~,
\eea
where $a_3>0$ is a function of $M$ and $\lambda$ is any 1-form on $M$ such that $\lambda(\partial_u)=1$.   It is straightforward to verify that by construction $M$ admits a null G-structure. For some null G-structures, like the Robinson structure described below, there is  a larger class of compatible metrics than the one described above.

\subsection{Symmetries of a null structure}

Let $M$ admit a $H_L^+(K)$-structure with fundamental  forms $\kappa$ and $\chi$.
A diffeomorphism $\Phi$
of $M$ preserves the null structure iff $\Phi^*\kappa= f_1\, \kappa$ and $\Phi^*\chi=f_2\, \chi$,
where $f_1$ and $f_2$ are spacetime functions which may depend on $\Phi$. In particular, the infinitesimal diffeomorphisms generated by a vector field $W$ preserve the null structure, iff
\bea
{\cal L}_W \kappa= a\,\kappa~,~~~{\cal L}_W\chi= b\, \chi~,
\label{sym1}
\eea
where $a,b$ are again functions of $M$.
\begin{prop}
Let $\chi=\kappa\wedge \omega$. If $W$ preserves the $H_L^+(K)$-structure on $M$ with fundamental forms  $\kappa$ and $\chi$, then
\bea
\hat{\cal L}_W\omega= p\, \omega
\eea
for some spacetime function $p$.
\end{prop}
\begin{proof}
Indeed
\bea
{\cal L}_W \chi={\cal L}_W \kappa\wedge \omega+ \kappa\wedge {\cal L}_W \omega={\cal L}_W \kappa\wedge \omega+ \kappa\wedge \hat{\cal L}_W \omega= (a+p) \chi~,
\eea
and so $p=b-a$.
\end{proof}

Apart from preserving the forms $\kappa$ and $\chi$, one may also require that $X$ is invariant under $W$ up to a conformal rescaling.
In such a case, an additional condition can be imposed as
\bea
[W,X]=c\, X~, \label{sym2}
\eea
for some spacetime function $c$.  The condition (\ref{sym2}) is not implied by the first condition in (\ref{sym1}) as an additional
restriction should be imposed on the metric which may not necessarily be   satisfied.

\subsection{Examples}

A large class of examples that includes many of the black holes and brane solutions in all dimensions can be described with the metric
\bea
g=-A^2(r) dt^2+ B^2(r) dr^2+ C^2(r) g(\Sigma)~,
\eea
where $t$ is the time coordinate, $r$ is a radial coordinate and $g(\Sigma)$ is the metric on a Riemannian manifold $\Sigma$.
To bring this metric into desirable form first change coordinates to $u^*=r^*, t=v\pm r^*$, where $dr^*= \pm B(r) A^{-1}(r) dr$, to yield
\bea
g=\mp 2 A^2(u^*) dv(du^*\pm {1\over2} dv)+C^2(u^*) g(\Sigma)~.
\eea
This has the desirable form but it often  convenient to introduce a coordinate $u$ such that  $du=\mp A^2(u^*) du^*$.  The metric then becomes
\bea
g=2 dv (du\pm {1\over2} D(u) dv)+ C^2(u) g(\Sigma)~.
\label{formmetr}
\eea
Clearly $X=\partial_u$, $\kappa=dv$ and $\lambda =du\pm {1\over2} D(u) dv$.  For the $n$-dimensional Schwarzschild  black hole with $n>3$,  one finds
$D(u)=\mp (1-2M u^{3-n})$, $C^2(u)=u^2$  and $g(\Sigma)=g(S^{n-2})$ is the round metric of $S^{n-2}$. Hence one has
\bea
g= 2 dv \left(du-{1\over2} (1-2M u^{3-n}) dv\right)+ u^2 g(S^{n-2})~.
\label{sch}
\eea
 While for the $n$-dimensional
Reissner–Nordstr\"om black hole with $n>3$,   one can show that
\bea
g=2 dv \left(du-{1\over2} {u^{n-2}-2M u+Q^2 u^{-n+4}\over u^{n-2}} dv\right)+ u^2 g(S^{n-2})~.
\label{rn}
\eea
Most metrics of brane and intersecting brane solutions in all dimensions have been brought into a similar   form to (\ref{formmetr})  in \cite{mbjffgp}
and the formulae will not be repeated here.

All the above solutions admit by construction a $H^+_L(SO(n-2))$-structure at least in the region of validity of the described coordinates.
Clearly all the examples above come with a chosen splitting of the tangent bundle.

The Kerr black hole solution admits a $H^+_L(U(1))$-structure.  In fact it is associated with a Robinson structure which we shall investigate in the next section.  This has been instrumental in the discovery of the solution \cite{kerrbh}. More examples of black hole solutions which can be shown to admit  null G-structures in various
dimensions can be found in \cite{mason}.

Another class of examples with a null G-structure are all supersymmetric solutions for which the isotropy group of a Killing spinor, $\epsilon$,
in the connected component of $Spin(n-1,1)$ is non-compact. All such solution are characterized by the existence of a nowhere vanishing null  Killing vector field $X$ such that $\slashed{X}\epsilon=0$, where $\slashed{X}$ is the Clifford algebra element associated to $X$.  Moreover the 1-form $\kappa$ is identified with the Dirac current of $\epsilon$,  $\kappa(X)=\langle\epsilon, \slashed{X}\epsilon\rangle_D=0$, where
$\langle\cdot, \cdot\rangle_D$ is the Dirac inner product. For all such backgrounds ${\cal L}_X \kappa=0$.  Furthermore,
 all the fundamental forms $\chi$ which are constructed as spinor bilinears of $\epsilon$ satisfy $\kappa\wedge \chi=0$, $i_X\chi=0$ and
 ${\cal L}_X \chi=0$.  Clearly $X$ generates a null geodesic congruence  in $M$. In addition  all such fundamental forms $\chi$ are invariant under $X$ and so they are the pull back of forms on the space of orbits ${\cal M}$ of the null geodesics.  No further conformal normalization of $\chi$ is necessary for this.

\section{Null $H^+_N(U(k))$- and $H^+_N(SU(k))$-structures }

\subsection{Null $H^+_N(U(k))$-structures and Robinson manifolds}

\subsubsection{Robinson manifolds}

Lorentzian manifolds with a $H^+_N(U(k))$-structure  are characterized by the existence
of two nowhere vanishing forms $\kappa$ and $\chi$, where $\chi$ after choosing a splitting and using proposition \ref{splitprop},  can be represented as $\kappa\wedge \omega$ with    $\omega$
an almost Hermitian form on $T$.  Furthermore $\omega$ and $g_T$ induce an almost complex structure on $T$.

On the other hand it is known for sometime that (almost) Robinson manifolds are closely associated to the existence
of an (almost) complex structure on a Lorentzian manifold.  To give the definition of these manifolds \cite{robinson}, see also \cite{Hughston:1988nz, nt},
recall that a  subspace $W$ of a Lorentzian vector space $V$ is null or totally null,  iff $W\cap W^\perp\not=\emptyset$ or $W\subset W^\perp$, respectively.  Furthermore $W$ is maximally totally null (MTN), iff $W=W^\perp$.

\begin{definition}
  Let $(M,g)$ be an even dimensional Lorentzian manifold. $(M,g)$ is an almost Robinson manifold, iff the complexified tangent bundle $TM\otimes\bC$ admits a subbundle $W$ whose fibre $W_p$ is  MTN subspace of $T_pM\otimes \bC$, $p\in M$.
\end{definition}

The definition of almost Robinson manifolds  makes some use of the spacetime metric which is used to define $N^\perp$. This is unlike the definition of
almost complex manifolds.  Of course almost complex manifolds are always almost Hermitian.  This is because   given an almost complex structure $I$, one can always find a Hermitian metric $h$ with respect to $I$, $h(U,V)=h(IU,IV)$,  by setting $h(U, V)=g(U,V)+g(IU, IV)$ for any metric $g$.
Nevertheless the definition of almost complex manifolds does not involve the metric.

\begin{prop}
Oriented and time oriented almost Robinson manifolds $(M,g)$  are Lorentzian manifolds with an $H^+_N(U(k))$ structure.
\end{prop}
\begin{proof}

A consequence of the definition of a Robinson manifold is that there is a real line bundle $N$ such that $W\cap \bar W=N\otimes \bC$ and $W+\bar W= N^\perp\otimes \bC$.  Moreover $N$ is topologically trivial as the spacetime is oriented and  time oriented. Furthermore, we have that
\bea
0\rightarrow N\otimes \bC\rightarrow W\rightarrow W/N\otimes \bC\rightarrow 0~,
\eea
and
\bea
0\rightarrow N\rightarrow N^\perp \rightarrow N^\perp /N\rightarrow 0~.
\eea
Suppose now that the two sequences above split and so in particular $W=N\otimes \bC\oplus R$ and  $N^\perp=N\oplus T$.   Moreover $R$ is an almost holomorphic subbundle of $T\otimes \bC$ and $T\otimes \bC=R\oplus \bar R$.  An almost complex structure is defined on $T$ as $I(v+\bar v)= i(v-\bar v)$, where $v$ is a section of $R$.  Next notice that the metric $g_T=g\vert_T$ on $T$ is Hermitian with respect to $I$.  This is because $W_p$ are MTN. As a result, one can define an almost Hermitian form, $\omega$,  on $T$, as $\omega(v_1, v_2)=g(v_1, I v_2)$.   In turn,    the 3-form $\chi=\kappa\wedge \omega$ is globally defined on $M$ and together with $\kappa(Z)=g(Z,X)$, where  $X$ is no-where vanishing section of $N$, define a $H^+_N(U(k))$ structure
on $M$.  The converse is also true.
\end{proof}

There is a notion of integrability of  almost Robinson structure in parallel to that of an almost complex structure, see \cite{robinson, Hughston:1988nz, nt}.
\begin{definition}
 $(M,g)$ is a Robinson manifold, iff $(M,g)$ is an almost Robinson manifold and
 \bea
[\Gamma(W), \Gamma(W)] \subset \Gamma(W)~,
\label{inconx}
\eea
where $[\cdot, \cdot]$ is the standard Lie bracket and $\Gamma(W)$ are the smooth sections of $W$.
\end{definition}

The condition (\ref{inconx})  is an integrability condition similar to that of the Nijenhouis condition for complex manifolds.  It is often convenient to express it in terms of forms.  For this define  the bundle $W^0$ whose fibres  are
$W_p^0=\{ \alpha\in T^*M\otimes\bC\vert ~\alpha(v)=0 ~\forall ~v\in W_p\}$, $p\in M$.
Then (\ref{inconx}) can also be expressed as
\bea
d\alpha= \beta\wedge \gamma~,
\eea
where $\alpha,\beta\in \Gamma(W^0)$ while $\gamma\in \Gamma(\Lambda^1(M))\otimes\bC$.  This integrability condition in terms of a local coframe
 $\{\fe^-, \fe^+, \fe^\alpha, \fe^{\bar \alpha}\}$  can be written as
\bea
d\fe^-=\fe^-\wedge \rho+ i h_{\alpha\bar\beta}\, \fe^\alpha\wedge \fe^{\bar\beta}~,~~~d\fe^{\bar\alpha}= \fe^-\wedge \mu^{\bar\alpha}+ \fe^{\bar\beta}\wedge \tau^\alpha{}_{\bar\beta}~,
\label{frameintcon}
\eea
where $\rho, \mu, \tau$ are spacetime 1-forms.

\begin{prop}
On Robinson manifolds the null vector field $X$ associated to $\kappa=\fe^-$  generates a null geodesic congruence.
\end{prop}
\begin{proof}
This has been demonstrated, see e.g. \cite{nt} and references within, and the proof will be repeated here for completeness. It suffices to show that (\ref{kappainv}) holds.  Indeed using the first condition in (\ref{frameintcon}), one finds that
\bea
{\cal L}_X \kappa=i_X d\kappa=- i_X\rho\, \kappa~,
\eea
and so $\kappa\wedge {\cal L}_X \kappa=0$. Then the result follows from proposition \ref{nullgeo}.
\end{proof}

It is well-known that many 4-dimensional solutions admit a Robinson structure.  In particular for the black hole solutions (\ref{sch}) and (\ref{rn}), one has
\bea
\fe^-=dv~,~~~\fe^{\bar 1}=\sqrt{2}\,u\, {d\bar z\over 1+z\bar z}~,
\eea
where $z$ is the inhomogeneous  complex coordinate on $S^2=\bC P^1$ and $\sqrt2$ appears for $S^2$  to have radius 1.   Observe that the integrability condition (\ref{frameintcon}) is satisfied.

\subsubsection{Geometry  of Robinson manifold null hypersurfaces}

Before we proceed to investigate the geometry of null  hypersurfaces
in Robinson spacetimes, it is helpful to state the definition of Cauchy-Riemann (CR) structures.

\begin{definition}
 A 2k+1-manifold ${\cal M}$ admits an almost Cauchy-Riemann (CR) structure, iff $T{\cal M}\otimes\bC$ admits a rank $k$ subbundle ${\cal W}$ such that at every point $p\in {\cal M}$,  ${\cal W}_p\cap \bar{\cal W}_p=\{0\}$. Moreover ${\cal M}$ admits a CR structure, iff in addition
 \bea
[\Gamma({\cal W}), \Gamma({\cal W})] \subset \Gamma({\cal W})~.
\label{incon}
\eea
\end{definition}
Any hypersurface in a complex manifold admits a CR structure\footnote{ It is not always the case that a manifold with a CR-structure
can be realized as a hypersurface of a complex manifold \cite{ln}. In modern terminology, the non-realizable CR manifolds  are not
``holographic'' in the context of complex geometry.  However those that are solutions of the vacuum Einstein equations are \cite{jlpnjt}. }.

\begin{prop}
Any null hypersurface ${\cal H}$ in a Robinson manifold transversal to the geodesic congruence generated by $X$ admits a CR structure.
\end{prop}
\begin{proof}
Let ${\bf n}$ be the normal to the hypersurface.  The spacetime metric $g$ restricted on ${\cal H}$ and $i_{\bf n}\chi$
define a fibre almost Hermitian structure on $T^{\cal H}$, $T{\cal H}={\cal N}\oplus T^{\cal H}$.  This can be used to decompose $T^{\cal H}\otimes \bC={\cal W}\oplus \bar{\cal W}$
where  the fibres of  ${\cal W}$ are the eigenvectors of the complex structure associated with the $i$ eigenvalue.
The integrability condition (\ref{incon}) also follows from that of the Robinson manifolds.  This is just the restriction of
(\ref{frameintcon}) on ${\cal H}$.
\end{proof}

\subsubsection{Geometry of orbit spaces}

Let ${\cal M}$ be the orbit space of null geodesics generated by $X$ in open subset $U$ of a Robinson manifold $M$ and $\pi$ the projection $\pi:~U\rightarrow {\cal M}$.

\begin{prop}
${\cal M}$ admits a CR-structure
\end{prop}
\begin{proof}
This has already been demonstrated in \cite{nt}. To show the statement  define the ${\cal W}=\pi_* W$ and $\bar{\cal W}=\pi_*\bar W$. These are subbundles of $T{\cal M}\otimes\bC$.
Moreover as $W\cap \bar W=N\otimes \bC$, $\pi_* X=0$ and the rank of $\pi_*$ is $2n-1$, we have that ${\cal W}_p\cap \bar{\cal W}_p=\{0\}$, $p\in {\cal M}$. Furthermore the
integrability of the sections of ${\cal W}$ follows from that of $W$ as $\pi_*[A,B]=[\pi_*A, \pi_*B]$.

\end{proof}

 Generically for $n>4$, ${\cal M}$ does not inherit   the fundamental forms associated with the Robinson structure on $U$ unless they are preserved by the flow of $X$. As $\kappa$ is preserved by the flow as a consequence of the integrability condition of the Robinson structure, there is a 1-form $\kappa_{\cal M}$ such that $a^{-1}\, \kappa=\pi^*\kappa_{\cal M}$ for some nowhere vanishing function $a$ on $M$. Similarly if $\chi=\kappa\wedge \omega$ is preserved
by the flow, there is a 3-form $\chi_{\cal M}$ such that $ b^{-1}\, \chi=\pi^* \chi_{\cal M}$, where $b$ is a strictly positive function on $U$.  As $\pi^*$ is an inclusion $\kappa_{\cal M}\wedge \chi_{\cal M}=0$, and so $\chi_{\cal M}=\kappa_{\cal M}\wedge \omega_{\cal M}$.  Furthermore if $g_T$ is preserved by the flow, one can define a fibre metric $g_{\cal M}$ on real part of ${\cal W}\oplus \bar{\cal W}$.  Then up to
an appropriate conformal rescaling of either $g_{\cal M}$ or $\omega_{\cal M}$ one can define a fibre hermitian structure such that
the holomorphic fibres are ${\cal W}_p$, see also section \ref{sec:gorbit}.

The $n=4$ case is special and both the fundamental forms $\kappa$ and $\chi$ are preserved by the flow generated by $X$ as a consequence
of the integrability conditions (\ref{frameintcon}).  As a result ${\cal M}$ apart from the CR structure also admits a globally defined 1- and 3-forms.

The Robinson manifolds $M$ for which all the fundamental forms and $g_T$ can be recovered from data on ${\cal M}$ up to appropriate
conformal rescaling locally can be written as $\bR\times {\cal M}$, where $X=\partial_u$ and $u$ a coordinate of $\bR$.  In particular,
the spacetime metric is
\bea
g= a\, \pi^*\kappa_{\cal M}\otimes \lambda+ a\, \lambda\otimes \pi^*\kappa_{\cal M}+ p\, \pi^* g_{\cal M}
\eea
where  $a , p>0$  are nowhere vanishing functions on the spacetime and $\lambda$ is any 1-form on the spacetime with $\lambda(\partial_u)=1$.  A more general construction of Robinson manifolds from the CR-structure on ${\cal M}$, $n>4$,  is given in \cite{nt}. Indeed if the integrability condition of the CR structure on ${\cal M}$ is expressed in terms of a local frame as in (\ref{frameintcon}), then the most general metric compatible with the induced Robinson structure on $M=\bR\times {\cal M}$ is
\bea
ds^2= a\, \pi^*\kappa_{\cal M}\otimes \lambda+ a\, \lambda\otimes \pi^*\kappa_{\cal M}+  q_{\alpha\bar\beta} \left(\pi^* e^\alpha \otimes \pi^*e^{\bar\beta}+ \pi^*e^{\bar\beta}\otimes \pi^* e^\alpha\right)~,
\eea
where $\overline{(q_{\alpha\bar\beta})}= q_{\beta\bar\alpha}$,  and both $a$ and $(q_{\alpha\bar\beta})$ depend on all coordinates.
However, the $\chi$ fundamental form on $M$ will  not always be the pull back of a form on ${\cal M}$.

We conclude this section on the Robinson structure with the remark that there are many refinements of the almost Robinson structure on $M$ as well as those induced on the hypersurfaces ${\cal H}$ and the orbit spaces ${\cal M}$.  These can be investigated with similar techniques to those developed  by Gray and Hervella \cite{grayhervella} for exploring the different   classes of almost Hermitian manifolds.  In particular, the covariant derivatives of the fundamental forms
 \bea
 \nabla \kappa~,~~~~\nabla \chi~,
 \label{gheqn}
 \eea
 where $\chi=\kappa\wedge^p\omega$, can be decomposed in irreducible representations of $H^+_L(U(k))$ or even $U(k)$, $n=2k+2$. The different classes are then identified depending on which components of (\ref{gheqn}) are non-vanishing in such a decomposition.  Some of these calculations have already been done in a similar Lorentzian context in  \cite{jguggp, het}, see also \cite{atc1,atc2}. Also several special structures can arise like  contact and Sasakian structures  on ${\cal H}$ and ${\cal M}$.

\subsection{Null $H^+_L(SU(k))$-structure}\label{sec:suk}

The $H^+_L(SU(k))$ structure apart from the two fundamental forms $\kappa$ and $\kappa\wedge \omega$ of the $H^+_L(U(k))$ structure also
exhibits a third fundamental form $\kappa\wedge\epsilon$, where $\epsilon$ is the holomorphic volume fundamental $(k,0)$-form of $SU(k)$, $n=2k+2$.
Furthermore, there are two relative normalization conditions that these forms satisfy
\bea
\kappa\wedge\omega\wedge\epsilon=0~,~~~ \kappa\wedge \omega^k= m(k)\,\kappa\wedge \epsilon\wedge \bar\epsilon~,
\label{2norm}
\eea
where $m(k)\not=0$ is a numerical normalization factor that depends on conventions whose precise value is not essential for the analysis that follows.
As the implications of the existence of $\kappa$ and $\kappa\wedge \omega$  on the geometry of the spacetime $M$ have already been investigated in the context of almost Robinson manifolds, the focus is on the properties of $\kappa\wedge \epsilon$.
It is clear that $\kappa\wedge \epsilon$ can be restricted on any null transversal  hypersurface ${\cal H}$.  This together with the two other fundamental forms induce a $\bR^+\times SU(k)$ structure on ${\cal H}$.

Furthermore if $\kappa$, $\kappa\wedge \omega$ and $\kappa\wedge \epsilon$ are preserved by the flow of $X$, they all, up to an appropriate
conformal rescaling, are the pull-back of appropriate forms on ${\cal M}$.  However this is not sufficient for ${\cal M}$ to admit a $H^+_L(SU(k))$-structure. For this
 the two normalization conditions (\ref{2norm}) must also be preserved by the flow.  In particular if ${\cal L}_X(\kappa\wedge \omega)=b\, \kappa\wedge \omega$ and ${\cal L}_X (\kappa\wedge \epsilon)= c\, \kappa\wedge \epsilon$, then for ${\cal M}$ to admit
a  $H^+_L(SU(k))$-structure, it is required that $k\,b=c+\bar c$.

\section{Symmetries of null structures}

To investigate the symmetries of the $H^+_L(K)$-structures, we shall first adapt a coordinate to $X$, $X=\partial_u$, and assume throughout that $\kappa\wedge d\kappa=0$. The latter implies that $\kappa=h dv$, for some coordinate $v$, where $h$ depends on all spacetime coordinates.  In such case, the most general spacetime vector field can be written as $W=W^u\partial_u+W^v \partial_v+ W^I \partial_I$,
where $y^I$ are the remaining spacetime coordinates.  In what follows, we shall identify the Lie algebra of $W$'s that preserve
a $H^+_L(K)$-structure for $K=SO(n-2)$ and $K=U(k)$, $n=2k+2$.

\subsection{Symmetries of the $H^+_L(SO(n-2))$ structure}

To identify the Lie algebra of  $W$s that preserve the $H^+_L(SO(n-2))$ structure, it seems reasonable to assume that $W$ preserve $X$, i.e. $[W,X]= a X$.  Then it is straightforward to deduce that
\bea
W=W^u(u,v,y) \partial_u+ W^v(v,y)\partial_v+ W^I(v, y) \partial_I~,
\label{invvf0}
\eea
where the dependence of vector fields $W$ on the various coordinates is implicitly  given.
From now on, we shall always take that the vector fields that generate the symmetries of the null structure are of the form (\ref{invvf0}).
Of course if the function $a$ vanishes, then $W^u=W^u(v,y)$.  However, we allow in what follows the component $W^u$ to depend on $u$.

Suppose that in addition $W$ preserves also $\kappa$, ${\cal L}_W\kappa= b \kappa$.  In such a case, one finds that
\bea
W=W^u(u,v,y) \partial_u+ W^v(v)\partial_v+ W^I(v, y) \partial_I~.
\label{invvf1}
\eea
Furthermore the invariance of the fundamental  (n-1)-form $\chi$ of $H^+_L(SO(n-2))$, ${\cal L}_W \chi= c\chi$, for a function $c$,
does not impose any additional conditions on $W$.  The component $W^v\partial_v$ generates the diffeomorphisms of a real line with coordinate $v$.  Let as denote the Lie algebra with $\mathfrak{diff}(\bR)$.  Similarly the component $W^I\partial_I$ generates the infinitesimal diffeomorphisms in the $y$ coordinates, which can be thought of as the coordinates of a manifold ${\cal S}$,  but now depend on $v$.  This can be thought as the Lie algebra of a path group of the diffeomorphism group of $y$'s which we denote with $\mathfrak{pdiff}({\cal S})$.
Similarly the component $W^u\partial_u$ can be thought as generating diffeomorphisms of the coordinate $u$ depending on both the $v$ and $y$ coordinates.  Denoting the Lie algebra of these with $\mathfrak{conf}(X)$, as they are a conformal rescaling of $X$, we have that the Lie algebra of symmetries is $\mathfrak{diff}(\bR)\oplus_s(\mathfrak{pdiff}({\cal S})\oplus_s \mathfrak{conf}(X))$, where $\oplus_s$ denotes semi-direct sum.

To get a bit more insight into these symmetries, suppose that the spacetime is  asymptotically flat  with a Killing horizon
at $u=u_h$, $g(\partial_v, \partial_v)\vert_{u=u_h}=0$. These conditions are met by the black hole solutions (\ref{sch}) and (\ref{rn}).  At the horizon hypersurface ${\cal H}$ the tangential component of the vector field (\ref{invvf1}) is
\bea
W^{\cal H}=W^v(v)\partial_v+ W^I(v, y) \partial_I~.
\eea
As $\partial_v$ is normal to ${\cal H}$, $\partial_v$ generates a null geodesic congruence in ${\cal H}$ whose space of orbits  is ${\cal S}$. ${\cal S}$ is the spatial horizon section of the horizon.  The Lie algebra of ${W^{\cal H}}$s  is $\mathfrak{diff}(\bR)\oplus_s\mathfrak{pdiff}({\cal S})$. If ${\cal S}$
is a sphere, ${\cal S}=S^{n-2}$, this algebra includes the Lorentz Lie algebra $\mathfrak{so}(n-1,1)$ as the conformal transformations of $S^{n-2}$.

On the other hand at infinity, the null hypersurface ${\cal I}$ which is transversal to the geodesic congruence generated by $X$ is given in all examples by the linear equation $pu+q v=0$, where $p,q\in\bR-\{0\}$, e.g. for the black hole solutions (\ref{sch}) and (\ref{rn}) $-2u+v=0$.  Adapting a coordinate $s$  along the asymptotic null infinity hypersurface ${\cal I}$, say $s=p u-q v$,
one finds that the component of (\ref{invvf1}) (and (\ref{invvf0})) tangential to ${\cal I}$ and  restricted on ${\cal I}$ read
\bea
W^{\cal I}=W^s(s, y)\partial_s+ W^I(s, y) \partial_I~.
\eea
The Lie algebra of $W^{\cal I}$s is the Lie algebra of the infinitesimal diffeomorphisms of ${\cal I}$, $\mathfrak{diff}({\cal I})$.     Therefore the spacetime diffeomeorphisms that leave the $H^+_L(SO(n-2))$-structure invariant interpolate between
 the symmetries of the horizon hypersurface  and those of the null asymptotic infinity.  Again if ${\cal S}=S^{n-2}$, $\mathfrak{diff}({\cal I})$ includes the BMS Lie algebra.  It is curious though that the symmetries that preserve the whole $H^+_L(SO(n-2))$-structure do not
 include the supertranslations of the BMS Lie algebra at the horizon.  Though of course this will be the case if instead of the vector fields in (\ref{invvf1}) one considers the vector fields in (\ref{invvf0}) which preserve part of the
 $H^+_L(SO(n-2))$-structure.

We have seen that the Lie algebras which preserve the $H^+_L(SO(n-2))$-structure are much larger than either the Lorentz or the BMS  algebra even when they are restricted at the asymptotic null infinity. This is not surprising as all considerations do not involve the spacetime metric
even at infinity.  However  one can attempt  to match the BMS  algebra at infinity for asymptotically flat spacetimes with that generated by
(\ref{invvf0}) and (\ref{invvf1}). This can be done in two different ways. First one can demand that the tangential components of (\ref{invvf0}) or (\ref{invvf1}) along ${\cal I}$ when restricted on ${\cal I}$ generate $\mathfrak{bms}$. Alternatively, one can demand the stronger condition that (\ref{invvf0}) or (\ref{invvf1}) preserve ${\cal I}$ and generate $\mathfrak{bms}$ when restricted on ${\cal I}$. The additional condition required for preserving ${\cal I}$ is that the normal component of $W$ should vanish when restricted on ${\cal I}$, $W^t\vert_{\cal I}=0$, where $t=pu+qv$. The vector fields that generate $\mathfrak{bms}$ on ${\cal I}$ are 
\bea
W^{\mathfrak{bms}}=(a(y) s+f(y)) \partial_s+ C(y)^I\partial_I~,
 \eea
 where $C^I(y)\partial_I$ generate conformal transformations of the sphere at infinity with infinitesimal conformal factor $2a$ and $f$ is an arbitrary function of the sphere at infinity. The tangential components of the vector fields (\ref{invvf0}) generate the $\mathfrak{bms}$ Lie algebra on ${\cal I}$ provided that
\bea
W={1\over p} \Big(  a(y, t)\, s+ f(y, t)+ q W^v(v, y)\Big) \partial_u+ W^v(y,v)\partial_v+ C^I(y)\partial_I~,
\eea
where $a(y)=a(y, t)\vert_{t=0}$, $f(y)=f(y, t)\vert_{t=0}$  and $t=p u+q v$. These vector fields retain their form under Lie brackets, i.e. close under Lie brackets,  provided that $a(y,t)=a(y)$ and $f(y,t)= f(y)$. While if the vector fields $W$ generate   $\mathfrak{bms}$ on ${\cal I}$ and preserve ${\cal I}$, then
\bea
W={1\over2p} \Big( a(y,t) s+ f(y,t)\Big ) \partial_u-{1\over 2q} \Big( a(y) (s-t)+ f(y)\Big ) \partial_v+ C^I(y) \partial_I~.
\eea
These vector fields retain their form under Lie brackets provided that $a(y,t)=a(y)$ and $f(y,t)=f(y)$.

Next turn to the vector fields in
(\ref{invvf1}). The tangential components of (\ref{invvf1}) will generate $\mathfrak{bms}$ on ${\cal I}$ provided that
\bea
W={1\over p}\Big(  a(y,t) \, s+ f(y,t)+q W^v\Big) \partial_u+ W^v \partial_v+ C^I(y) \partial_I~,
\label{invvf3}
\eea
where $W^v=W^v(v)$. These vector fields close under Lie brackets  provided that $a(y,t)=a(y)$ and   $f(y,t)=f(y)$.
It turns out that there do not exist vector fields (\ref{invvf1}) which both generate $\mathfrak{bms}$ on ${\cal I}$ and preserve ${\cal I}$.

\subsection{Symmetries of the integrable $H^+_L(U(k))$ structure}

As in the investigation of symmetries of the $H^+_L(SO(n-2))$-structure, we assume that $\kappa\wedge d\kappa=0$ and adapt coordinates $u,v$
such that $X=\partial_u$ and $\kappa=h dv$.  We have already established that the transformations that preserve either $X$ or both $X$ and $\kappa$
are generated by the vector fields (\ref{invvf0}) and (\ref{invvf1}), respectively.  Next instead of imposing that $W$ also leave invariant the fundamental form $\kappa\wedge \omega$, we impose that the condition that $W$ leaves invariant the complex structure $I$ in $N^\perp/N$, $\check {\cal L}_WI=0$.  Assuming that  there are coordinates $(u,v, z^\alpha, z^{\bar\alpha})$ on the spacetime such that $I^\alpha{}_\beta=-I^{\bar \alpha}{}_{\bar\beta}=i \delta^\alpha{}_\beta$,    a straightforward  computation reveals that the condition $\check{\cal L}_WI=0$ for the vector fields  (\ref{invvf0}) implies
that
\bea
W=W^v( y,v)\partial_v+W^u(u, v, z, \bar z) \partial_u+W^\alpha(v,z) \partial_\alpha+ W^{\bar \alpha}(v,\bar z) \partial_{\bar\alpha}~.
\label{invvf6a}
\eea
while for the vector fields (\ref{invvf1}) implies
\bea
W=W^v(v)\partial_v+W^u(u, v, z, \bar z) \partial_u+W^\alpha(v,z) \partial_\alpha+ W^{\bar \alpha}(v,\bar z) \partial_{\bar\alpha}~.
\label{invvf6}
\eea
It is significant that the components of $W^\alpha\partial_\alpha$ depend only on the holomorphic coordinates $z$.  Therefore
for each $v$ these generate the holomorphic diffeomorphisms on  ${\cal S}$.  The Lie algebra of the above vector fields (\ref{invvf6})
is $\mathfrak{diff}(\bR)\oplus_s(\mathfrak{phol}({\cal S})\oplus_s \mathfrak{conf}(X))$, where $\mathfrak{phol}({\cal S})$ is the Lie algebra of the path group of the holomorphic diffeomorphisms of ${\cal S}$.

A similar argument to that we have used for $H^+_L(SO(n-2))$ reveals that the tangent components of the vector fields  (\ref{invvf6}) along  a Killing horizon hypersurface ${\cal H}$ are
\bea
W^{\cal H}=W^v(v)\partial_v+ W^\alpha(v,z) \partial_\alpha+ W^{\bar \alpha}(v,\bar z) \partial_{\bar\alpha}~,
\eea
while those along an asymptotic null hypersurface ${\cal I}$ are
\bea
W^{\cal I}=W^s(s,z, \bar z)\partial_s+W^\alpha(z) \partial_\alpha+ W^{\bar \alpha}(\bar z) \partial_{\bar\alpha}~.
\eea
For $n>6$-dimensional black holes with ${\cal S}=S^{n-2}$, the above vector field do not generate either the Lorentz group at ${\cal H}$ or
the BMS group at ${\cal I}$.  This is because $S^{2k}$, $k>3$, do not admit (almost) complex structures-the problem for $k=3$ remains open.

However for 4-dimensional black holes, the Lie algebra of $W^{\cal H}$ vector fields is $\mathfrak{diff}(\bR)\oplus_s \mathfrak{psl}(2,\bC)$, where $\mathfrak{psl}(2,\bC)$ is the Lie algebra of the path group of $SL(2,\bC)$.  The latter are the holomorphic transformations of the
2-sphere acting with  M\"obius transformations on the complex coordinate of $\bC P^1=S^2$. As $\mathfrak{sl}(2,\bC)=\mathfrak{so}(3,1)$, the Lorentz Lie algebra is included.  A similar analysis reveals that $W^{\cal I}$
generate the Lie algebra $\mathfrak{pbms}$ where we have used again that $\mathfrak{sl}(2,\bC)=\mathfrak{so}(3,1)$.  

Furthermore one can match the vector fields (\ref{invvf6a}) and  (\ref{invvf6}) with  those generating by the BMS  group at ${\cal I}$. The analysis is similar to that we have presented  in section 5.1 with
\bea
C^I\partial_I=W^\alpha(z) \partial_\alpha+ W^{\bar \alpha}(\bar z) \partial_{\bar\alpha}~.
\label{invvf7}
\eea
Of course one can extend the BMS group by including the singular holomorphic transformations of $S^2$ as in \cite{barnich}.

\section{Conclusions}

We have presented a definition of null G-structures on Lorentzian manifolds which generalizes both the Robinson structure and the null structures which arise
in the context of supersymmetric backgrounds in  supergravity theories.  Then we proceed to explore the geometric properties
of spacetimes admitting such null G-structures utilizing the fundamental forms that characterize  the null G-structures.
We have also examined   the relationship between   null G-structures and  null geodesic congruences.  Then we use this to explore  the geometry of null   hypersurfaces transversal to null geodesic congruences as well as that of the orbit spaces of null geodesics.
We have investigated in more detail the null $H^+_L(K)$-structures associated with the groups  $K=SO(n-2)$, $U(k)$ and $SU(k)$, $n=2k+2$, of a n-dimensional spacetime. We have established that Robinson manifolds admit a  $H^+(U(k))$ structure.
Furthermore, we have examined the symmetries of some null G-structures and demonstrated that interpolate between the symmetries
of Killing horizons and those of null asymptotic infinity. The symmetry algebra of a $H^+_L(K)$-structure for $K=SO(n-2)$,  and $K=U(1)$ and $n=4$, on asymptotically flat spacetimes includes the BMS  algebra.

We have established that many solutions of four and higher dimensional gravitational theories admit a null G-structure. These include the Schwarzchild and  Reissner–Nordstr\"om black holes in four and higher dimensions. It also includes by construction all supersymmetric solutions which
admit a null vector spinor bilinear. It is expected that many more solutions can admit a null G-structure. However to identify which one  would require an extensive search amongst the plethora of solutions that have been constructed specially in the context
of string  and M- theories.

Supersymmetric solutions with a null vector bilinear admit a null G-structure everywhere on the spacetime.  This is because a
 Killing spinor has to be defined everywhere on the spacetime, away from singularities, for a solution to be supersymmetric. In non-supersymmetric solutions though the null G-structures are typically defined on open, but not necessarily small, subsets of the spacetime. It will be of interest to explore the patching conditions
that are required for a null G-structure to be extend across  the whole of spacetime.

The different Gray-Hervella type of  classes of a null $H^+_L(K)$-structure can be  explored  in the same way as those for almost Hermitian manifolds \cite{grayhervella}.  In fact it may be convenient to decompose the covariant derivatives $\nabla\kappa$ and $\nabla\chi$ of the fundamental forms $\kappa$ and $\chi$ in terms of the representations
of the topological structure group $K$ instead of $H^+_L(K)$. Some such calculations in the Lorentzian case have already been done in \cite{jguggp, het}, see also \cite{atc1,atc2}. The various solutions of gravitational theories with a null $H^+_L(K)$-structure will belong
to one of the classes that will arise.

%%%%%%%%%%%%%%%%%%%%%%%%%%%%%%%%%%%%%%%%%%%%%%%%%%%%%%%%%%%%%%%%%%%%%%%%%%%%%%%%%%%%%%%%%%%%%%%%
\section*{Acknowledgments}

  I would like to thank  C.~ Bachas for an invitation to visit \'Ecole Normale Superi\'eure,  and  for hospitality and a stimulating environment    to complete this project. I  am partially supported from the  STFC rolling grant ST/P000258/1.

%%%%%%%%%%%%%%%%%%%%%%%%%%%%%%%%%%%%%%%%%%%%%%%%%%%%%%%%%%%%%%%%%%%%%%%%%%%%%%%%%%%%%%%%%%%%%%%%
%\setcounter{section}{0}\setcounter{equation}{0}

%\appendix{Notation and conventions}

\end{document}